\documentclass{article}
\usepackage{color,xcolor}
\usepackage{amsmath}
\usepackage{amsfonts}
\usepackage{amssymb}
\usepackage{amsthm}
\usepackage{enumerate}
\usepackage{fullpage}
\usepackage{bm}
\usepackage{thmtools, thm-restate}
\usepackage[nottoc]{tocbibind}

\newtheorem{theorem}{Theorem}[section]
\newtheorem{lemma}[theorem]{Lemma}

\newtheorem{remark}[theorem]{Remark}

\newtheorem*{theorem*}{Theorem \ref{thm:main}}
\newtheorem*{corollary1}{Corollary \ref{cor:result}}

\newtheorem{definition}[theorem]{Definition}
\newtheorem{corollary}[theorem]{Corollary}

\newcommand{\F}{\mathbb{F}}

\newcommand{\Z}{\mathbb{Z}}

\newcommand{\poly}{\mathsf{poly}}

\newcommand{\Fam}{\mathcal{F}}
\newcommand{\U}{\mathcal{U}}
\newcommand{\V}{\mathcal{V}}
\newcommand{\vecX}{\mathbf{X}}
\newcommand{\mgroup}{H_m}
\newcommand{\mgen}{\gamma_m}

\newcommand{\Otilde}{\Tilde{O}}

\def\anon{0}

\ifnum\anon=0
\newcommand{\fnote}[1]{{\color{brown} [Fatemeh: #1]}}
\newcommand{\mnote}[1]{{\color{red} [Madhu: #1]}}
\newcommand{\snote}[1]{{\color{blue} [Swastik: #1]}}
\else  
\newcommand{\fnote}[1]{}
\newcommand{\mnote}[1]{}
\newcommand{\snote}[1]{}
\fi 

\usepackage{parskip}

\title{Improved PIR Schemes using Matching Vectors and Derivatives} 

\ifnum\anon=0
{
\author{Fatemeh Ghasemi\thanks{Department of Mathematics, University of Toronto, Canada. Email: \texttt{fatemeh.ghasemi@mail.utoronto.ca}}
\and 
Swastik Kopparty\thanks{Department of Mathematics and Department of Computer Science, University of Toronto, Canada. Research supported by an NSERC Discovery Grant.
Email: \texttt{swastik.kopparty@utoronto.ca}}
\and 
Madhu Sudan\thanks{School of Engineering and Applied Sciences, Harvard University, Cambridge, Massachusetts, USA. Supported in part by a Simons Investigator Award and NSF Award CCF 2152413. Email: \texttt{madhu@cs.harvard.edu}}
}}
\else
\author{}
\fi 

\begin{document}

\maketitle

\begin{abstract}

In this paper, we construct new $t$-server Private Information Retrieval (PIR) schemes with communication complexity subpolynomial in the previously best known, for all but finitely many $t$. 
Our results are based on combining derivatives (in the spirit of Woodruff-Yekhanin~\cite{WY}) with the Matching Vector based PIRs of Yekhanin~\cite{Yekhanin} and Efremenko~\cite{Efremenko}. Previously such a combination was achieved in an ingenious way by Dvir and Gopi~\cite{DG}, using polynomials and derivatives over certain exotic rings,  en route to their fundamental result giving the first $2$-server PIR with subpolynomial communication. 
\\
\\
Our improved PIRs are based on two ingredients:
\begin{itemize}
\item We develop a new and direct approach to combine derivatives with Matching Vector based PIRs. This approach is much simpler than that of Dvir-Gopi: it works over the same field as the original PIRs, and only uses elementary properties of polynomials and derivatives.
\item A key subproblem that arises in the above approach is a higher-order polynomial interpolation problem. We show how ``sparse $S$-decoding polynomials", a powerful tool from the original constructions of Matching Vector PIRs, can be used 
to solve this higher-order polynomial interpolation problem using surprisingly few higer-order evaluations.
\end{itemize}
Using the known sparse $S$-decoding polynomials from~\cite{Efremenko,IS,Cheeetal} in combination with our ideas leads to our improved PIRs. Notably, we get a $3$-server PIR scheme with communication $2^{\Otilde( (\log n)^{1/3}) }$, improving upon the previously best known communication of $2^{\Otilde( \sqrt{\log n})}$ due to Efremenko~\cite{Efremenko}.


\end{abstract}

\newpage 

\section{Introduction}
Private Information Retrieval (PIR), introduced in \cite{CGKS} by Chor, Kushilevitz, Goldreich and Sudan, is a method for a user to interact with $t$ non-colluding servers and read some part of a database without revealing to the servers which part of the database was read. Specifically, a database $(a_1, \ldots, a_n) \in \{0,1\}^n$ is known to all the $t$ servers, and the user wants to find the value of $a_\tau$ without revealing any information about $\tau$ to any server. PIRs are typically studied with $t$ being a constant, and we will only consider this case.  There is a trivial protocol which works even for $t=1$: some server just sends the entire database to the user -- this uses $n$ bits of communication. With more servers it is (surprisingly) possible to use much less communication. The main problem, which has been extensively studied in the subsequent years, is to design PIR schemes with as little total communication as possible.

The first generation of PIR schemes were based on Reed-Muller codes, also known as multivariate polynomial evaluation codes, over finite fields $\F$. We give a quick taste of the most elementary such PIR scheme. The original data $a_1, \ldots, a_n$  is used to specify a multivariate polynomial $F(X_1, \ldots, X_k) \in \F[X_1, \ldots, X_n]$ of degree $t-1$ in $k = n^{1/(t-1)}$ variables by ensuring that the value of $F$ at a certain point $u_i \in \F^k$ equals $a_i$. When the user wants to access $a_\tau$, it chooses a uniformly random line $\ell$ through $u_\tau$, 
picks points $b_1, \ldots, b_{t}$ on $\ell$, and asks server $j$ for the value of $F$ at $b_j$. Using the fact that uniformly random lines through $u_i$ cover the space $\F^k$ uniformly, we get that each $b_i$ is uniformly distributed, and thus no server learns anything about $\tau$. Using the fact that the restriction $F|_{\ell}$ of the low degree multivariate polynomial $F$ to the line $\ell$ is a low degree univariate polynomial, we get that the univariate polynomial $F|_{\ell}$ can be completely recovered from its values at the $b_i$ -- and thus $a_\tau$, the value of $F$ at $u_\tau$ can be recovered by the user. This gives a PIR scheme with $O(k) = O(n^{1/(t-1)})$ communication, and is nontrivial for $t \geq 3$.

With more sophisticated ideas,~\cite{CGKS} gave a $2$-server PIR scheme with communication $O(n^{1/3})$, and even lower communication for $t$-server PIR for $t \geq 3$. Further improvements to this were given by~\cite{Amb97,BI01,BIKR02}, leading to a
$n^{O(\frac{\log\log t}{t \log t })}$-communication protocol for large $t$.

Then in 2007, the breakthrough result of Yekhanin~\cite{Yekhanin} completely changed the landscape by giving $3$-server PIR schemes with subpolynomial ($n^{O(\frac{1}{\log\log n})}$) communication (assuming the infinitude of Mersenne primes). Soon after, building on a greatly clarifying presentation and reinterpretation of~\cite{Yekhanin} by Raghavendra~\cite{Raghavendra}, a beautiful result of Efremenko unconditionally achieved
significantly smaller communication $2^{\Otilde(\sqrt{\log n})}$ for $3$-server PIR, and even lower communication (of the form $2^{\Otilde( (\log n)^{\varepsilon_t} ) }$ for more servers).

There are two key ingredients in these constructions: $S$-Matching Vector Families (SMVF), and Sparse $S$-Decoding polynomials (SSD) for a set $S \subseteq \Z_m$ with $0 \in S$. $S$-Matching Vector Families are collections of vectors $(u_i, v_i) \in \Z_m^k \times \Z_m^k$ with certain restrictions on the inner products $\langle u_i, v_j \rangle \in \Z_m$: all $\langle u_i,v_j\rangle$ should lie in $S$, with $\langle u_i, v_j \rangle = 0$ if and only if $i = j$. The challenge is to get as large a collection as possible. Sparse $S$-Decoding polynomials are sparse polynomials that take certain prescribed values at the set of points $ \{ \gamma^{s}: s \in S \}$, where $\gamma$ is a primitive $m$th root of $1$ in some field.  The challenge is to get the sparsity (= number of monomials) as small as possible. Yekhanin's original construction took $m$ being prime, and was based on moderately large $S$-Matching Vector Families but astonishingly sparse $S$-decoding polynomials. Efremenko's construction was based on $m$ being composite, and was based on very large $S$-Matching Vector Families (originating in the work of Grolmusz~\cite{Gro} and Barrington-Beigel-Rudich~\cite{BBR}) and slightly nontrivially sparse $S$-decoding polynomials.
Subsequent work by Itoh-Suzuki~\cite{IS} and Chee-Feng-Ling-Wang-Zhang~\cite{Cheeetal} improved the sparsity of $S$-decoding polynomials, which led to quasipolynomial reduction in the communication over Efremenko's result) for $t$-server PIR for all $t \geq 9$.

We refer to this entire approach as the SMVF+SSD framework. We will give a detailed overview of this framework soon.

In the other direction, there are only very weak lower bounds known for PIR. Wehner and de Wolf \cite{WdW} showed a $5 \log n$ lower bound on communication for $2$-server PIR. It is possible (but would be very surprising if so) that there are $2$-server PIR schemes with $O(\log n)$ communication.

PIR schemes are closely related to Locally Decodable Codes (LDCs), and progress on both has often come together. Indeed,
Reed Muller codes are the original LDCs, and the SMVF+SSD framework is also the engine behind the record holding Matching Vector Code constant query LDCs.
In the non-constant query regime, local decoding is still very interesting,
and Matching Vector Codes are LDCs in this setting too -- this was shown by Dvir-Gopalan-Yekhanin~\cite{DGY} and Ben-Aroya-Efremenko-Ta-Shma~\cite{BET} -- whose work also shed further light on the SMVF + SSD framework.
For the formal connections between PIR schemes and LDCs, see Katz-Trevisan~\cite{KatzTrevisan} and Trevisan~\cite{Trevisan04}. 


One basic question that remained open for a while was whether 2-server PIR could be done with subpolynomial communication. 
This was resolved by the beautiful work of Dvir-Gopi~\cite{DG}, which is also the starting point for our work.
Dvir-Gopi developed a way to reduce the number of servers in Matching Vector Family based PIRs by increasing the amount of communication being sent by the servers. This required an ingenious new adaptation of an idea of Woodruff-Yekhanin~\cite{WY}, originally for the setting of Reed-Muller code based PIRs, to the setting of Matching Vector Family based PIRs. The Woodruff-Yekhanin idea in the setting of the Reed-Muller based PIR scheme described earlier, is to make the servers also return all the higher order partial derivatives of the multivariate polynomial $F$ at all the queried points $b_i$. This higher order evaluation information for the multivariate polynomial $F$ will let the user deduce, via the chain rule, higher order evaluations of the univariate degree $t-1$ polynomial $F|_{\ell}$ from each server-- and leads to less servers being needed overall. To get an analogue of this in the setting of Matching Vector Family based PIRs, Dvir-Gopi worked over an exotic ring $R_m = \Z_m[\gamma]/\langle\gamma^m - 1\rangle$ in place of a base finite field, and developed some theory of derivatives and interpolation from higher order evaluations over $R_m$. Not all facts from the finite field case carry over, but enough facts did, and this enabled the reduction of the number of servers needed for subpolynomial communication from $3$ to $2$.

For $2$ servers this gives PIR schemes with $2^{\Otilde(\sqrt{\log n})}$ communication.
The method of Dvir-Gopi  also generalized to a larger number of servers, and led to communication complexity subpolynomial in the previous best bounds, whenever the number of servers $t$ lies in $\{2\} \cup \{4,5,6,7,8\} \cup \{16, \ldots, 23\}$. (Specifically the previous communication complexity were quasipolynomial, but not polynomial, in the new communication complexity.) 


For the remaining $t$ (which includes all $t\geq 27$), the previous PIR schemes of Efremenko instantiated with the sparse $S$-decoding polynomials of~\cite{IS,Cheeetal} remained the best known.

In this paper, we construct $t$-server PIR schemes with communication complexity that is subpolynomial in the previously best known bound, for all but finitely many $t$ (in fact, our improvements are for precisely the $t$ where the Dvir-Gopi scheme did not improve the state of the art).  Notably, we get a $3$-server PIR scheme with communication $2^{\Otilde( (\log n)^{1/3}) }$, improving upon the previously best known communication of $2^{\Otilde( \sqrt{\log n})}$ due to Efremenko~\cite{Efremenko}. A key part of our improvement is a new and simpler method to apply derivatives (in the spirit of Woodruff-Yekhanin) to reduce the number of servers in Matching Vector based PIRs -- and in particular giving a simpler proof of the Dvir-Gopi results. This method works natively over the finite field $\F_p$ and only uses elementary properties of polynomials over finite fields. 

As a by product of our simpler approach to applying derivatives, we are able to formulate a higher multiplicity analogue of the problem of having a sparse $S$-decoding polynomial. Miraculously, it turns out that whenever there is an unusually sparse $S$-decoding polynomial in the original sense, it automatically translates into an unusually sparse higher multiplicity
analogue of the $S$-decoding polynomial -- thus leading to an immediate reduction, using the existing results on sparse $S$-decoding polynomials from~\cite{Efremenko,IS,Cheeetal}, in the number of required servers (beyond the improvement that derivatives already gave).

We believe our result is the more natural way to use derivatives to improve the communication in the SMVF+SSD framework. The total communication of our new PIR scheme can be very succinctly stated as follows: {\em If the SMVF+SSD framework yields a $t$-server PIR scheme with communication $\exp(\Otilde((\log n)^{\frac1r})$, then we get a $t$-server PIR scheme 
with communication $\exp(\Otilde((\log n)^{\frac{1}{r+1}}))$.}

In terms of concrete parameters, the net result is that we improve the communication complexity to a subpolynomial in the previously best known, for the state of the art $t$-server PIR for all $t$ except $t \in \{2\} \cup \{4,5,6,7,8 \} \cup \{16, 17, \ldots, 23 \}$ (for these exceptions our bound is no better than that of Dvir and Gopi, although our protocol and the proof of correctness is simpler).

For concrete comparison, 
we give a table listing the number of servers needed to achieve $\exp(\Otilde((\log n)^{1/c}))$ communication for different integers $c$. We list this below 
(for a few small values $c$) for our PIR scheme, the Dvir-Gopi scheme and the original Efremenko scheme (using the best known $S$-decoding polynomials from ~\cite{Efremenko,IS,Cheeetal}).

\begin{tabular}{|c|c|c|c|}
\hline
    Total communication     &   \# servers for our PIR scheme & \# servers for \cite{DG}  & \# servers for \cite{Efremenko} \\
    \hline
   $\exp(\Otilde((\log n)^{1/2}))$ &   2 & 2 & 3 \\
   $\exp(\Otilde((\log n)^{1/3}))$ &   3 & 4 & 8 \\
   $\exp(\Otilde((\log n)^{1/4}))$ &   8 & 8 & 9 \\
   $\exp(\Otilde((\log n)^{1/5}))$ &   9 & 16 & 24 \\
   $\exp(\Otilde((\log n)^{1/6}))$ &   24 & 32 & 27 \\
   $\exp(\Otilde((\log n)^{1/7}))$ &   27 & 64 & 72 \\
   $\exp(\Otilde((\log n)^{1/8}))$ &   72 & 128 & 81 \\
   $\exp(\Otilde((\log n)^{1/9}))$ &   81 & 256 & 216 \\
   \hline
\end{tabular}

\subsection{Overview}

We now give a quick introduction to PIR schemes based on the SMVF+SSD framework, discussing the important ideas of Yekhanin~\cite{Yekhanin}, Raghavendra~\cite{Raghavendra}, Efremenko~\cite{Efremenko} and Dvir-Gopalan-Yekhanin~\cite{DGY}.

Our presentation is slightly more algebraic than usual.  For the expert,
the main point is that it treats $F( \cdot)$, the multivariate function encoding the data, as a genuine polynomial with natural number exponents, and not as an exponential polynomial with $\Z_m$ exponents. We then consider its restriction to a parametrized curve $C(Z)$, leading to a univariate polynomial $A(Z) = F(C(Z))$ of quite large degree. It is only at this stage where we restrict ourselves to evaluations at $m$th roots of unity, by reducing the univariate polynomial $A(Z)$ mod $Z^m -1$.
This change in viewpoint will greatly help when considering derivatives: it is the key reason our PIR scheme and its proof are simpler than that of Dvir-Gopi, and also underlies our improved PIR schemes.

Let $\F$ be a field, and let $(a_1, \ldots, a_n) \in \F^n$ be the database for which we want a PIR scheme. The key ingredient that we need is a matching vector family.
For an integer $m$ and a set $S \subseteq \Z_m$ with $0 \in S$, an {\em $S$-matching vector family}
is a collection of pairs $(u_1, v_1), \ldots, (u_n, v_n)$, where each $u_i$ and $v_i$
lies in $\Z_m^k$ such that the inner products obey the following restrictions:
$$ \langle u_i, v_j \rangle \in \begin{cases}
    \{0\} & i = j \\
    S \setminus \{0 \} & i \neq j
\end{cases}.$$
We will also view the $u_i$ and $v_i$ as vectors in $\{0,1,\ldots, m-1\}^k \subseteq \mathbb N^k$. A key fact is that whenever $m$ is composite, superpolynomial size matching vector families exist.

Next we choose a field $\F$ of characteristic $p$ relatively prime to $m$ which contains 
a set $\mgroup$ of all $m$ $m^{\mathrm{th}}$ roots of unity.

The database is used to specify a multivariate polynomial $F(\mathbf X) \in \F[X_1, \ldots, X_k]$ as follows:
$$ F(\mathbf X) =  \sum_{i = 1}^{n} a_i \mathbf X^{u_i}.$$
The $u_i$, which are a-priori in $\Z_m^k$, are treated here as elements of $\mathbb N^k$. This is part of the change in viewpoint that will help us later.

Now if the user wants to recover $a_{\tau}$, it will use $v_\tau$ to pick a random parametrized curve $C(Z) \in (\F[Z])^k$ (a $k$-tuple of univariate polynomials)
as follows: for uniformly random ${\bm \beta} = (\beta_1, \ldots, \beta_k) \in \mgroup^k$, it sets $C(Z) = (\beta_1 Z^{v_{\tau,1}}, \ldots, \beta_k Z^{v_{\tau, k}})$. 

Because of the randomness of ${\bm \beta}$, for any $h \in \mgroup$, $C(h)$ is distributed uniformly at random in $\mgroup^k$. Thus asking each server for the value of $F$ at a single $C(h)$ will maintain privacy perfectly. This will give us access to values of the composed polynomial $A(Z) = F(C(Z))$ at some points $h \in \mgroup$.
    
The key point in the definition of $C$ is that $A(Z)$ has the following nice form:
$$A(Z) = \sum_{i=1}^{n} a_i {\bm \beta}^{u_i} Z^{\langle u_i, v_\tau \rangle},$$
where the inner products $\langle u_i, v_\tau \rangle$ are treated as inner products of nonnegative integer vectors (and thus $A(Z)$ is of huge degree $\approx k \cdot m^2$).

To utilize the matching vector property, which talks about the inner products
$\langle u_i, v_\tau\rangle$ mod $m$ and not the integer $\langle u_i, v_\tau \rangle$, we will only substitute $Z$ to be an element of $\mgroup$.
Equivalently, we can talk about the remainder of $A(Z)$ mod the polynomial $Z^m -1$, 
which has the effect of reducing all the exponents of $Z$ mod $m$. This is the view we take, and this is the point of difference alluded to earlier.


Define $A_1(Z) = A(Z) \mod (Z^m - 1)$.
By the matching vector property and the above discussion, 
$A_1(Z)$ is of the form:
$$A_1(Z) = a_\tau {\bm \beta}^{u_\tau} + \sum_{s \in S \setminus \{0\}} c_s Z^s,$$
and thus the $a_\tau$ we seek can be found from the constant term of $A_1(Z)$.
Further, since the polynomial $Z^m-1$ has all the elements of $\mgroup$ as roots, the values of $A_1$ on $\mgroup$ equal the values of $A$ on $\mgroup$. Thus the user gets access to $A_1$ at some points $h \in \mgroup$.

What remains for the user is to solve an interpolation problem: {\em given evaluations of a polynomial of the form $A_1(Z)$ at some points in $\mgroup$, find the constant term of $A_1$}. The number of evaluations needed determines the number of servers - and this is where sparse $S$-decoding polynomials come in -- sometimes this interpolation problem can be solved using surprisingly few evaluation points.

\paragraph{Using Derivatives:}
Now we will see how to improve this using derivatives.

Following the idea of Woodruff-Yekhanin in the setting of Reed-Muller code based PIR schemes,
Dvir and Gopi~\cite{DG} suggested asking the servers for not only the value of $F$ at points $C(h)$, but also for higher order derivatives of $F$ at these points. The hope was that fewer higher-order evaluations are needed to solve the interpolation problem, and thus fewer servers are needed for the PIR scheme. This created several technical complications. Most significantly, taking derivatives of a polynomial whose exponents are in $\Z_m$ seemed to require (since $(Z^i)' = i\cdot Z^{i-1}$ -- the $i$ came down from the exponent to the coefficient) the coefficient ring of the polynomials to accommodate elements of $\Z_m$. This was handled by Dvir and Gopi through an ingenious idea -- to replace the coefficient field $\F$ for the polynomials with a {\em ring} $R_m$ which by design contains the $m$th roots of $1$ and also $\Z_m$; namely $R_m = \Z_m[\gamma]/\langle \gamma^m -1 \rangle$.

The use of the exotic ring $R_m$ enabled many aspects of this idea to go through; however the interpolation properties of values and derivatives do not behave as nicely as they do over fields. Here~\cite{DG} tackle the resulting interpolation problem hands on, and the final proof ultimately relies on nonvanishing (but not noninvertibility!) of some mysterious $4\times 4$ determinants over $R_m$ in the $2$-server case, and an explicit solving of a certain linear system over $R_m$ by reducing to rings of the form $\Z_p[\gamma]/\langle \gamma^m - 1\rangle$ and the Chinese remainder theorem for the general $t$-server case. 

We remark that some further simplifications are possible in the $2$-server case.
For $2$-servers and when $m = 6$, Dvir-Gopi find a homomorphism from the ring $R_6$ to the field $\F_3$ which preserves the vital nonvanishing determinant in their proof -- thus the entire protocol can be made to work directly over $\F_3$. For a larger number of servers there does not seem to be an analogue of this.
In a recent striking result, Alon, Beimel and Lasri~\cite{ABL2024} find an extremely elementary protocol and proof of the $2$-server PIR result of~\cite{DG}, without ever mentioning polynomials (but still using Matching Vector Families). This is also only for 2 servers.

We now describe our approach to using derivatives to improve the communication of the matching vector PIR schemes. It works over the original field $\F$ itself and uses only classical derivatives over fields.
We will set $M = mp$ and take the matching vector family $(u_i, v_i)$ over the ring $\Z_{M}$.
Then again we encode the data $(a_1, \ldots, a_n) \in \F^n$ in the polynomial 
$$ F(\mathbf X) = \sum_i a_i {\mathbf X}^{u_i}.$$
Since $M$ has one more prime factor than $m$, the size of the matching vector family that we take can be quasipolynomially larger than what can be done in the original Efremenko scheme.
We will still only be evaluating $F$ at points of $\mgroup$, the $m^{\rm{th}}$ roots of $1$.

Now suppose the user wants to recover $a_\tau$.
It will again pick a random curve $C(Z)$ just as before:
$C(Z) = ( \beta_1 Z^{v_{\tau,1}}, \ldots, \beta_k Z^{v_{\tau,k}})$,
where ${\bm \beta} \in \mgroup^k$.
Because of the randomness of ${\bm \beta}$, for any $h \in \mgroup$, $C(h)$ is distributed uniformly at random in $\mgroup^k$. Thus asking each server for the value and higher order derivatives of $F$ at a single $C(h)$ will maintain privacy perfectly. Via the chain rule for derivatives, this will give the user access to values and higher order derivatives of the the composed polynomial $A(Z) = F(C(Z))$ at some points $h \in \mgroup$.

Again, the key point in the definition of $C$ is that $A(Z)$ has the following nice form:
$$A(Z) = \sum_{i=1}^{n} a_i {\bm \beta}^{u_i} Z^{\langle u_i, v_\tau \rangle}.$$
To utilize the matching vector property, which talks about the inner products
$\langle u_i, v_\tau\rangle$ mod $M$ and not the integer $\langle u_i, v_\tau \rangle$, we will reduce $A(Z)$ mod $(Z^M - 1)$.


Define $A_1(Z) = A(Z) \mod (Z^M - 1)$.
By the matching vector property and the above discussion, 
$A_1(Z)$ is of the form:
$$A_1(Z) = a_\tau {\bm \beta}^{u_\tau} + \sum_{s \in S \setminus \{0\}} c_s Z^s,$$
and thus the $a_\tau$ we seek can be found from the constant term of $A_1(Z)$.

Now for the crucial point. Since we are in characteristic $p$, 
$$Z^M - 1 = Z^{mp}-1^p = (Z^{m} - 1)^p,$$
which vanishes at each of the $m$ points of $\mgroup$ with multiplicity $p$.

Thus the evaluations  of $A_1(Z)$ {\em and its first $p-1$ derivatives} agree with the
evaluations of $A(Z)$ and its first $p-1$ derivatives at the points of $\mgroup$.
Since we have access to these higher order evaluations of $A(Z)$ at the points of $\mgroup$, we now find ourself faced with a higher-order interpolation problem:
{\em Given higher order evaluations of a polynomial of the form $A_1(Z)$ at some points in $\mgroup$, find the constant term of $A_1$.} The number of higher order evaluations determines the number of servers, and it seems reasonable to hope that we will need fewer higher order evaluations to find the constant term (because we get more information from a higher order evaluation).

This brings us to the final ingredient. We want a small set of points in $\mgroup$ which suffice for the above higher order interpolation problem. We show how to get unusually small such sets from unusually small sets of points which suffice for the original interpolation problem -- namely, from sparse $S$-decoding polynomials. This is again done using some algebraic insights; most crucially that the order $p$ evaluation of $A_1(Z)$ at the point $b$ is completely determined by $A_1(Z) \mod (Z-b)^p$,
and using the identity $(Z-b)^p = Z^p -b^p$ gives us concrete handle on $A_1(Z) \mod (Z-b)^p$.

\section{Main Result}

In order to state our main theorem, we need two quick definitions to set up the infrastructure of the SMVF+SSD framework.

\begin{definition}[Canonical set]
\label{def:canonicalset}
Let 
$m \in \Z$
be a positive integer. We define the canonical set for 
$m$
to be the following set:
$$S = \{x \in \Z_m \ s.t. \ x^2 \equiv x \ mod \ m\}$$
\end{definition}

\begin{definition}[$S$-decoding polynomial]
Let $S \subseteq \Z_m$ be a set containing $0$,
Let $\F$ be a field containing an element $\gamma_m$
which is a primitive $m$'th root of $1$: namely
$\gamma_m^m = 1$ 
and 
$\gamma_m^i \neq 1$
for 
$i = 1, 2,..., m - 1 $.

A polynomial 
$P(Z) \in \F[Z]$
is called an 
$S$-decoding polynomial if the following conditions hold:
\begin{itemize}
    \item $\forall s \in S \setminus \{0\}: \ P(\gamma_m^s) = 0$
    \item $P(\gamma_m^0) = P(1) = 1$
\end{itemize}
\end{definition}

\begin{theorem}[Main Theorem]
\label{thm:main}
    Let $m$ be a positive integer with $r$ distinct prime factors, and let $p$ be a prime not dividing $m$.

    Suppose $\F$ is a field of characteristic $p$ such that, for $S^* \subseteq \Z_m$ being the canonical set for $m$, there is an $S^*$-decoding polynomial over $\F$ with sparsity $\leq t$.

    Then there is a $t$-server PIR scheme with communication $\exp(\Otilde( ( \log n)^{\frac{1}{r+1}}))$.
\end{theorem}
The proof appears in Section~\ref{sec:proofs-of-main-thms}.

For comparison, the main result of Efremenko concerning PIRs is below: it takes the exact same hypothesis and produces a PIR scheme with quasipolynomially larger communication.

\begin{theorem}[\cite{Efremenko}]
    Let $m$ be a positive integer with $r$ distinct prime factors, and let $p$ be a prime not dividing $m$.

    Suppose $\F$ is a field of characteristic $p$ such that, for $S^* \subseteq \Z_m$ being the canonical set for $m$, there is an $S^*$-decoding polynomial over $\F$ with sparsity $\leq t$.


    Then there is a $t$-server PIR scheme with communication $\exp(\Otilde( ( \log n)^{\frac{1}{r}}))$.
\end{theorem}

Instantiating Theorem~\ref{thm:main} with arbitrary $m$ and $p$, and using the fact that every $S$ has a trivial $S$-decoding polynomial over $\F_p$ of sparsity $|S|$, we recover the main result of Dvir and Gopi giving $2^{r}$-server PIR schemes with communication $\exp(\Otilde( (\log n)^{\frac{1}{r+1}}))$ for all $r \geq 1$.

Instantiating Theorem~\ref{thm:main} with known $m,p$ for which unusually sparse $S^*$-decoding polynomials are known~\cite{Efremenko,IS,Cheeetal}, we get new PIR schemes better than previously known. In particular, we get a $3$-server PIR scheme with $\exp(\Otilde((\log n)^{1/3}))$ communication, improving upon the $\exp(\Otilde((\log n)^{1/2}))$ communication for $3$-server PIR first proved by Efremenko~\cite{Efremenko}.

For larger $t$, this improves the state of the art communication for $t$-server PIR for all $t$ except\footnote{For these exceptional $t$, the PIR scheme of Theorem~\ref{thm:main} matches the communication of the Dvir-Gopi scheme, while being simpler to describe and analyze.} $t \in \{2 \} \cup \{ 4, 5, \ldots, 8 \} \cup \{ 16, 17, \ldots, 23 \}$ . Formally, the parameters we get for $t$-server PIR using the best known $S$-decoding polynomials from~\cite{Cheeetal} are given by the following corollary.
\begin{corollary}
\label{cor:result}
    For each $r \geq 2$, there is a $t$-server PIR with communication
    $\exp(\Otilde(\log n)^{\frac{1}{r+1}})$ according to the following table:
    \begin{center}
\begin{tabular}{ |c|c| } 
 \hline
 r & t  \\ 
 \hline
 $r = 1$ & 2\\
 $r$ even, $r \leq 102$ & $\left(\sqrt{3}\right)^r$ \\ 
 $r$ odd, $3 \leq r\leq 103$  & $8 \cdot \left(\sqrt{3} \right)^{r-3} $ \\ 
 $r \geq 104$  &  $\left( 3/4 \right)^{51} \cdot 2^r$ \\ 
 \hline
\end{tabular}
\end{center}
\end{corollary}
The proof appears in Section~\ref{sec:proofs-of-main-thms}.

\section{Preliminaries}
We use
$\F_q$
to denote a finite field of
$q$
elements and
$\Z_m$
to denote the ring of integers modulo
$m$.
The inner product between two vectors
$u = (u_1,..., u_k), v = (v_1,..., v_k)$
is denoted by
$\langle u, v \rangle = \sum_{i = 1}^k u_i \cdot v_i$.


The pointwise product, $\odot$, of two vectors
$u = (u_1, \ldots, u_k)$ and $v = (v_1, \ldots, v_k)$,
denoted $u \odot v$ is
the vector $(u_1 \cdot v_1, u_2 \cdot v_2, \ldots, u_k \cdot v_k)$.

\begin{definition}[$S$-Matching Vector Family]
Let
$S \subset \Z_m, 0 \in S$
and
$\Fam = \left( \U, \V \right)$
where 
$\U = \left(u_1,..., u_n \right), \V = \left(u_1,..., u_n \right)$ with  $u_i, v_i \in \Z_m^k$.
Then 
$\Fam$
is said to be an 
$S$-matching 
vector family of size 
$n$
and dimension
$k$
if the following conditions hold:
\begin{itemize}
    \item $\langle u_i, v_i \rangle = 0$ for every $i \in [n]$
    \item $\langle u_i, v_j \rangle \in S \setminus \{0\}$ for every $i \neq j$
\end{itemize}
\end{definition}

\begin{remark}
We changed the definition to have $0 \in S$ (traditionally $0$ is not included), as we felt it made things more convenient notationally.

In particular, our definition makes all values of
$\langle u_i, v_j \rangle$, where $i,j \in [n]$,
lie in $S$.
\end{remark}

\begin{theorem}[Large Matching Vector Families, Theorem 1.4 in \cite{Gro}]
\label{thm:MVF}
Let 
$m = p_1p_2...p_r$ where $p_1, p_2,..., p_r$ are distinct
primes and $r \geq 2$. Then, there exists an explicitly constructible $S$-matching vector family 
$\Fam$ 
in
$\Z_m^k$
of size 
$n \geq exp\left( \Omega\left(\frac{(log \ k)^r}{(log \ log \ k)^{r - 1}} \right)  \right)$
where 
$S$
is the canonical set for $m$.
\end{theorem}

\begin{theorem}[Sparse $S$-Decoding Polynomials, Theorem 4.1 in \cite{Cheeetal}]
\label{thm:decoding-poly}
Let
$r$
be a positive integer. Then, there exists an integer 
$m = p_1...p_r$
, product of 
$r$
distinct odd primes, such that the 
$S$-decoding polynomial for 
$m$
has the following number of monomials where  
$S$ 
is the canonical set for
$m$
, and the corresponding field has characteristic
$2$.
\begin{itemize}
    \item $r = 1$: \\
         number of monomials = $2$
    \item $2 \leq r \leq 103$: \\
        number of monomials = $\begin{cases}
        (\sqrt{3})^r & \text{if $r$ is even} \\
        8 \cdot (\sqrt{3})^{r - 3} & \text{if $r$ is odd}
          \end{cases}$
    \item $r \geq 104$:\\
        number of monomials = $(\frac{3}{4})^{51}\cdot2^r$
\end{itemize}
\end{theorem}



We now define (Hasse) derivatives of polynomials and recall some important properties. Let 
$\F$
be a field and 
$\F[\vecX] = \F[X_1,..., X_k]$
be the ring of polynomials. Also, for a vector 
$i = (i(1),..., i(k))$
of non-negative integers, denote its weight by
$wt(i) = \sum_{j = 1}^k i(j)$
and denote the monomial
$\prod_{j = 1}^k X_j^{i(j)}$
by
$\vecX^i$.

\begin{definition}[Hasse derivatives]
 For a multivariate polynomial 
 $F(\vecX) \in \F[\vecX]$
 and a non-negative vector
 $i = (i(1),..., i(k)),$
we denote the
$i$-th
Hasse derivative of 
$F$
by
$F^{(i)}(\vecX)$
, and it is the coefficient of 
$\mathbf{Y}^{i}$
in the polynomial 
$F(\vecX + \mathbf{Y}) \in \F[\vecX, \mathbf{Y}]$
. Therefore, we have,
$$F(\vecX + \mathbf{Y}) = \sum_i F^{(i)}(\vecX)\mathbf{Y}^i$$
\end{definition}

We use the notation 
$F^{(<m)}(a)$
to be the vector of all the evaluations of 
$i$-th
Hasse derivative of 
$F$
at point
$a$
for all vectors 
$i$
such that
$wt(i) < m$.
We note that Hasse derivatives are just a simple scaling of standard iterated derivatives whenever the order of derivative is smaller than the characteristic of the field. In particular, for the univariate polynomial $A(Z) = Z^s$, we have a simple formula for the Hasse derivative: $A^{(j)}(Z) = {s \choose j} Z^{s-j}$.

We will need the following facts about Hasse derivatives.
\begin{itemize}
    \item Let $A(Z) \in \F[Z]$ be a univariate polynomial, and let $b \in \F$.
    Then knowing $A(Z) \mod (Z-b)^m$ is exactly equivalent to knowing $A^{(<m)}(b)$.
    
    Concretely, given
    $A^{(<m)}(b)$
    we can construct the polynomial
    $$A^*(Z) = \sum_{i = 0}^{m - 1} A^{(i)}(b) (Z - b)^i,$$
    which is easily seen to be the remainder of 
    $A(Z) \mod (Z - b)^m$. This connection goes in the other direction too.
    \item Chain rule: Suppose $F(X_1,..., X_k) \in \F[X_1,..., X_k]$ be a multivariate polynomial and $C = (P_1(Z),..., P_k(Z))$
    be a tuple of univariate polynomials in the single variable $Z$. Denote their composition by the polynomial
    $A(Z) = F(C(Z))$
    . Then, the $i$-th Hasse derivative of 
    $A$
    , 
    $A^{(i)}$
    , can be expressed as a polynomial in 
    $F^{(j)}(C(Z))$
    for 
    $wt(j) \leq i$
    and $P_t^{(j)}(Z)$
    for $1 \leq t \leq k, j \leq i$.
\end{itemize}

\begin{definition}[PIR Scheme]
A one-round t-server PIR protocol involves t servers, 
$\mathcal{S}_1,..., \mathcal{S}_t$, each holding the same 
$n$-bit 
database
$a = (a_1,..., a_n) \in \{0, 1\}^n$
, such that these 
$t$
servers do not communicate with each other,
and there is a user 
$\mathcal{U}$
who knows 
$n$
and wants to retrieve some bit 
$a_i$
for
$i \in [n]$
, without
revealing 
$i$. It consists of a randomized algorithm for the user, and
$t$
deterministic algorithms for the servers such that:
\begin{itemize}
    \item On input
    $i \in [n]$,
    $\mathcal{U}$ 
    obtains a random string $rand$ and produces 
    $k$ random queries 
    $q_1,..., q_t$
    and send them to the respective servers.
    \item Each server $j$ produces a response 
    $r_j = \mathcal{S}_j(a, q_j)$ such that $a$ is the database, and send it back to the user.
    \item The user based on the randomness, $i$, $r_1,..., r_t$ calculates $a_i$, the $i$ bit of the database.
\end{itemize}

The protocol should satisfy the following conditions:
\begin{itemize}
    \item $\mathbf{Correctness}: $ For any database $a$, index $i$, user should output the value of $a_i$ with probability $1$, where the probablity is over the random string $rand$.
    \item  $\mathbf{Privacy}: $ Each server individually learns nothing about $i$ i.e. for any fixed database $a$, and for any server, the distributions of 
    $q_j(i_1, rand), q_j(i_2, rand)$ for all $i_1, i_2 \in [n]$ are identical.
\end{itemize}
\end{definition}

\section{The $0$-interpolation property, with multiplicity}

In this section, we express the property of having a sparse $S$-decoding polynomial 
in terms of having a small interpolating set for a certain interpolation problem. We then define a multiplicity version of that interpolation problem that will be useful for the improved PIR scheme, and give a general way to get small interpolating sets for this high multiplicity interpolation problem from small interpolating sets for the original interpolation problem without multiplicities (and thus from sparse $S$-decoding polynomials).

\begin{definition}[$0$-interpolation property]
    Let $S \subseteq \mathbb N$ with $0 \in S$. We say $B \subseteq \F$ has the $0$-interpolating property for $S$ if there is a map $E: \F^B \to \F$,
    such that for every polynomial $R(Z)$ of the form
    $$ R(Z) = \sum_{s \in S} c_s Z^s,$$
    we have $E(R|_B) = R(0)$.
    In words, the evaluations of $R$ at the points of $B$ determine $R(0)$.
\end{definition}

\begin{definition}[$0$-interpolation property with multiplicity $e$]
    Let $S \subseteq \mathbb N$ with $0 \in S$. We say $B \subseteq \F$ has the $0$-interpolating property with multiplicity $e$ for $S$ if 
    there is a map $E : (\F^e)^B \to \F$ such that for every polynomial $R(Z)$ of the form
    $$ R(Z) = \sum_{s \in S} c_s Z^s,$$
    we have $E(R^{(<e)}|_B) = R(0)$. 
    In words, the order $e$ evaluations of $R$ at the points of $B$ determine $R(0)$.
\end{definition}

\begin{lemma}[Having sparse $S$-decoding polynomials is the same as having a small set with the $0$-interpolation property]
\label{lem:polytointerpolate}
    Suppose $p$ is a prime and $m$ is relatively prime to $p$. Let $\F$ be a field of characteristic $p$ containing all $m$ $m$th roots of $1$. Let $\mgroup$ denote this set of $m$th roots of $1$.
    
    Then the following are equivalent:
    \begin{itemize}
        \item There is an $S$-decoding polynomial over $\F$ for $\Z_m$ with $ \leq t$ monomials.
        \item There exists a subset $B$ of $\mgroup$ with $|B| \leq t$ which has the $0$-interpolation property for $S$.
    \end{itemize}
\end{lemma}

\begin{proof}
    Suppose $P(Y) = \sum_{j=1}^t e_j Y^{d_j}$ is an $S$-decoding polynomial over $\F$ for $\Z_m$.
    So $P(\mgen^s) = 0$ for $s \in S \setminus \{0\}$, and $P(1) = 1$.

    Define $B = \{ \mgen^{d_j} \mid j \in [t] \}$.

    Let $R(Z)$ be a polynomial given by:
    $$ R(Z) = \sum_{s \in S} c_s Z^s.$$

    Then we can use the following expression for computing $c_0$ from $R|_B$:
    \begin{align*}
        \sum_{j \in [t]} e_j R(\mgen^{d_j}) &=  \sum_{j \in [t]} e_j \sum_{s \in S} c_s \cdot (\mgen^{d_j})^ s\\
        &= \sum_{s \in S} c_s \sum_{j \in [t]} e_j \cdot (\mgen^s)^{d_j} \\
        &= \sum_{s \in S} c_s P(\mgen^s)\\
        &= c_0 P(\mgen^0) + \sum_{s \in S \setminus \{0\}} c_s P(\mgen^s)\\
        &= c_0.
    \end{align*}
    This implies that $B$ has the $0$-interpolation property for $S$.

    Now we show the other direction.
    Suppose $B \subseteq \mgroup$ has the $0$-interpolation property for $S$.
    Let $E: \F^B \to \F$ be the corresponding map.
    First observe that $E$ must be a linear map; this is because for polynomials $R_1(Z), R_2(Z)$ we have $(R_1 + R_2)|_B = (R_1)|_B + (R_2)|_B$ and $(R_1 + R_2)(0) = R_1(0) + R_2(0)$.
    
    This means that any polynomial $R(Z)$ of the form $\sum_{s \in S} c_s Z^s$ that vanishes on all points in $B$ must also vanish at $0$. 

    Let the elements of $B$ be $b_1 = \mgen^{d_1}, \ldots,  b_t = \mgen^{d_t}$.
    Define the $[t] \times S$ matrix $W$ with $(j, s)$ entry equal to $b_j^s = \mgen^{sd_j}$. Expressing the above information in terms of $W$ we get that for any vector $\mathbf c \in \F^S$, if $W \mathbf c = 0$, then we must have $w_0 \cdot \mathbf c = 0$, 
    where $w_0 \in \F^S$ is the vector with $1$ in coordinate $0$ and $0$ in all the remaining coordinates $S \setminus \{0\}$.

    This means that the rows of $W$ span $w_0$.
    Suppose $e \in \F^t$ is such that $e W = w_0$. Then it is easy to check that the polynomial
    $$P(Y) = \sum_{j = 1}^{t} e_j Y^{d_j}$$
    is an $S$-decoding polynomial over $\F$ for $\Z_m$.
\end{proof}

\begin{lemma}[$0$-interpolation property with multiplicity from the $0$-interpolation property without multiplicity]
\label{lem:SmtoSM}
    Suppose $p$ is a prime and $m$ is relatively prime to $p$.  
    Let $M = mp$.
    Let $\phi: \Z_m \times \Z_p \to \Z_M$ be the Chinese remainder isomorphism.

    Suppose $S_m \subseteq \Z_m$, $S_M \subseteq \Z_M$ and $e \in \{1, 2, \ldots, p\} \in \mathbb N$ satisfy:
    $$ 0 \in S_M \subseteq \phi( S_m \times \{0,1,\ldots, e-1\} ).$$

    Let $\F$ be a field of characteristic $p$ containing all $m$ $m$th roots of $1$. Let $\mgroup$ denote this set of $m$th roots of $1$.

    Suppose $B \subseteq \mgroup$ is a $0$-interpolating set for $S_m$. 
    Then $B$ is a $0$-interpolating set of multiplicity $e$ for $S_M$.
\end{lemma}

\begin{proof}
    Let $R(Z)$ be a polynomial of the form $\sum_{s \in S_M} c_s Z^s$.

    We first prove this for $e = p$.

    The key fact that we will use is that for an element $b \in \F$, specifying $R^{(<p)}(b)$ is exactly equivalent to specifying the remainder $R(Z) \mod (Z - b)^p$.


    Then $R(Z) \mod (Z-b)^p = R(Z) \mod (Z^p - b^p)$ can be computed as:
    \begin{align*}
    &\sum_{s \in S_M}   c_s   ( Z^s \mod (Z^p - b^p))\\
    &= \sum_{s \in S_M}  c_s (Z^{s \mod p} \cdot  b^{p \cdot \lfloor s/p \rfloor}) \\
    &= \sum_{s'' \in [e]} \left( \sum _{s' \in S_m}  c_{\phi(s',s'')} b^{p \cdot \lfloor \phi(s',s'')/p \rfloor}              \right) Z^{s''}\\
    &= \left( \sum _{s' \in S_m}  c_{\phi(s',0)} b^{p \cdot \lfloor \phi(s',0)/p \rfloor}    \right) Z^{0} + \sum_{s'' \in \{1, \ldots, e-1\}} \left( \sum _{s' \in S_m}  c_{\phi(s',s'')} b^{p \cdot \lfloor \phi(s',s'')/p \rfloor}              \right) Z^{s''}
    \end{align*}
    Note that $\phi(s',0)$ is a multiple of $p$ in $\Z_M$, and thus $p \cdot \lfloor \phi(s', 0)/p \rfloor = \phi(s',0)$. Now $\phi(s',0) \equiv s' \mod m$, and thus for any $b \in H$, we have $b^{\phi(s',0)} = b^{s'}$.
    Thus, if we define $R_1(Z)$ by:
    $$R_1(Y)  = \sum_{s' \in S_m} c_{\phi(s',0)}Y^{s'},$$ the constant term of $R(Z) \mod (Z-b)^p$ equals:
    $$ \sum_{s' \in S_m} c_{\phi(s',0)} b^{\phi(s',0)} = \sum_{s' \in S_m} c_{\phi(s',0)} b^{s'}= R_1(b).$$
    Thus the order $p$ evaluations of $R$ at $B$ determines the evaluations of $R_1$ at $B$.
    By the $0$-interpolation property of $B$ for $S_m$, we know that $R_1|_B$ determines the constant coefficient of $R_1$, namely $c_{\phi(0,0)} = c_0$.
    This completes the proof for $e = p$.

    To get the result for any $e \leq p$, we note that $R^{(j)}(Z) = 0$ for any $j \in \{e, e+1, \ldots, p-1\}$. This is because our hypothesis about $S_M$ means that for any $s \in S_M$, $s$ is either $0,1, \ldots$ or  $e-1 \mod p$ -- and so the $j$'th Hasse derivative of $Z^s$, ${s \choose j} Z^{s-j}$, is $0$ mod $p$ by Lucas's theorem.
    This means that the derivatives of order $\geq e$ are not providing any information. Thus knowledge of $R^{(<e)}(b)$ is equivalent to knowledge of $R^{(<p)}(b)$, and thus the previous result implies that we may just take order $e$ evaluations instead of order $p$ evaluations to recover $c_0$.
\end{proof}

\section{The PIR scheme}

In this section, we give our improved PIR scheme.

To specify the PIR scheme, we need:
\begin{itemize}
    \item A prime $p$ and a natural number $m$ relatively prime to $p$.  Define $M = mp$.
    \item An integer $e$ with $1 \leq e \leq p$ -- this will govern the multiplicity of interpolation we need.
    \item A finite field $\F$ of characteristic $p$ containing the set $\mgroup$ of $m$ different $m$th roots of $1$.
    \item A set $S_M \subseteq \Z_M$.
    \item An $S_M$-Matching Vector Family $(u_i, v_i) \in (\Z_M^k)^2$ for $i \in [n]$.
    \item An set $B = \{b_1, \ldots, b_t\} \subseteq \mgroup$ with the $0$-interpolating property with multiplicity $e$ for $S_M$ (over $\F$).
\end{itemize}

 When we instantiate this for a given $k$, we would like $B$ to be small (this governs the number of servers) and $n$ to be big (this governs the size of the database that can be handled with $\poly(k)$ communication). As we will see, we will eventually instantiate this by taking an 
 $S_m \subseteq \Z_m$  for which there is a sparse $S_m$-decoding polynomial over $\F$, and then take $S_M = S_m \times \{0,1\} \subseteq \Z_m \times \Z_p \simeq \Z_M$ -- and fortunately there happens to be a large $S_M$-Matching Vector Family for this set $S_M$.

Suppose the user is interested in finding the value of
$a_\tau$. The PIR scheme works as follows:
\begin{enumerate}
    \item The user picks a random ${\bm \beta} \in \mgroup^k$.
    \item Define the parametrized curve $C$ to be the tuple of polynomials:
    $$ C(Z) =   {\bm \beta} \odot Z^{v_\tau} \in \F[Z]^k.$$
    Here $v_\tau$ is viewed as a vector in $\{0,1,\ldots, M-1\}^k \subseteq \mathbb N^k$.
    Observe that $C$ maps $\mgroup$ to $\mgroup^k$.
    For every $i \in \{1,\ldots,t\}$, the user sends $C(b_i) \in \mgroup^k$ to server $i$.
    \item 
        The servers define a polynomial 
        $F(X_1,..., X_k) \in \F[X_1,..., X_k]$
    representing the data as follows:
    $$F(X_1,..., X_k) = \sum a_i {\mathbf X}^{u_i}$$
    Here the $u_i$ are viewed as vectors in $\{0,1,\ldots, M-1\}^k \subseteq \mathbb N^k$.
    
    Each server evaluates $F^{(<e)}$ at the point in $\mgroup^k$ that it received, and returns
    that to the user.
    Thus server $i$ sends $$F^{(<e)}(C(b_i)) \in \F^{{k + e-1 \choose e-1}}.$$
    \item The user will now use all the received information to deduce $a_\tau$.

    \paragraph{The plan:}
    Define $A(Z) = F(C(Z))$, which is a polynomial of degree at most $(M-1)^2k$.
    Define $A_1(Z) = A(Z) \mod (Z^{M} - 1)$.
    Observe that $A_1(Z)$ has degree $< M$ and is of the form:
    $$ {\bm \beta}^{u_{\tau}}a_\tau +  \sum_{s \in S\setminus \{0\}} w_s Z^s.$$ 
    (Here we view each $s \in S \subseteq \Z_M$ as an integer in $\{0, \ldots, M-1\}$.

    The user will find $A_1(0) = {\bm \beta}^{u_{\tau}} a_\tau$, from which $a_\tau$ can be computed.

    To do this, using the fact that $B$ is a $0$-interpolating set with multiplicity $e$, it suffices for the user to find $(A_1^{(<e)}(b_i))_{i = 1}^t$.

    How do we find $A_1^{(<e)}(b_i)$? We now note that it is the same as $A^{(<e)}(b_i)$.
    Indeed, since $(Z-b_i)$ divides $Z^m - 1$, we have that $(Z - b_i)^p$ divides $(Z^m-1)^p = (Z^M - 1)$, and so:
    $$A^{(<p)}(b_i) = A_1^{(<p)}(b_i),$$ and in particular, $$A^{(<e)}(b_i) = A_1^{(<e)}(b_i).$$

    Finally $A^{(<e)}(b_i)$ can be computed by the user from $F^{(<e)}(C(b_i))$ and knowledge of $C$.

    \paragraph{The implementation:}

    From the chain rule, the user can take the servers' answers, $F^{(<e)}(C(b_i))$, and knowledge of $C^{(<e)}(b_i)$,
    to compute     $A^{(<e)}(b_i)$ for each $i \in [t]$.

    By the earlier observation, $A^{(<e)}(b_i) = A_1^{(<e)}(b_i),$ and so the user knows order $e$ evaluations of $A_1$ at each point of $B$. 

    Finally, by the $0$-interpolation property with multiplicity $e$ of $B$, the user can recover $A_1(0) = {\bm \beta}^{u_\tau} a_\tau$, 
    and thus (since all coordinates of $\beta$ are nonzero) compute $a_\tau$, as desired.
\end{enumerate}
This concludes the description of the PIR scheme.


\begin{theorem}
\label{thm:PIR-general}
The protocol described above gives a $t$-server PIR scheme for a database of size $n$ with the communication cost
$O(k + k^{e-1})$.


\end{theorem}
\begin{proof}
The above procedure shows that the user can recover the value of 
$a_\tau$
for any
$\tau \in [n]$
with probability 
$1$. In the protocol, the user sends a vector in 
$\F^k$
, and each server sends a vector in 
$\F^{{k + e-1 \choose e-1}}.$ 
Therefore, the total communication is 
$O(k + k^{e-1})$. For the privacy, note that 
${\bm \beta}$
is a uniformly random vector in 
$\mgroup^k,$ independent of $\tau$. 
Hence for every $i$, the query 
$C(b_i)$
has the same distribution  for every pair of possible message indices $\tau$ and $\tau'$. 
\end{proof}

\subsection{Instantiating the PIR Scheme: proof of the main theorem}
\label{sec:proofs-of-main-thms}

We are ready to prove Theorem \ref{thm:main}.

\begin{theorem*}
   Let $m$ be a positive integer with $r$ distinct prime factors, and let $p$ be a prime not dividing $m$.

    Suppose $\F$ is a field of characteristic $p$ such that, for $S^* \subseteq \Z_m$ being the canonical set for $m$, there is an $S^*$-decoding polynomial over $\F$ with sparsity $\leq t$.


    Then there is a $t$-server PIR scheme with communication $\exp(\Otilde( ( \log n)^{\frac{1}{r+1}}))$.
\end{theorem*}
    
\begin{proof}
By our assumption about $\F$, we know that $\F$ is a finite field of characteristic $p$ which contains all $m$ $m$th roots of $1$ (i.e., $m | (|\F|-1)$). Let $\mgroup \subseteq \F$ denote this set of $m$th roots of $1$ in $\F$.

Let 
$M = mp$
, and let $S_M$ be the canonical set (recall Definition~\ref{def:canonicalset}) for $M$.
Use Theorem \ref{thm:MVF} to  obtain an
$S_M$-matching vector family $\Fam$ of size $n$ and dimension $k$, where 
$k = \exp(\Otilde( ( \log n)^{\frac{1}{r+1}}))$
, and $r + 1$ comes from the fact that $M$ is the product of $r + 1$ distinct primes. 

By Lemma \ref{lem:polytointerpolate} applied to the $S^*$-decoding polynomial given by our hypothesis, we get a subset 
$B \subseteq \mgroup$, with $|B| \leq t$, which has the
$0$-interpolation property for 
$S^*$.

We will now apply Lemma \ref{lem:SmtoSM} to $B$ to get that it is a $0$-interpolating set with multiplicity $e = 2$. The key point is that the canonical set for $M$, namely $S_M$, is a cartesian product of the canonical set for $m$, namely $S^*$, under the Chinese remainder isomorphism.
Indeed:
$$0 \in S_M \subseteq \Z_M,$$
and that
$$\phi(S^* \times \{0, 1\}) = S_M$$
where $\phi$ is the Chinese remainder isomorphism $\phi : \Z_m \times \Z_p \to \Z_M$
and $S^*$ is the canonical set for 
$m$.
Thus by Lemma~\ref{lem:SmtoSM} with $e = 2$, we get that
$B$ 
is a 
$0$-interpolating set of multiplicity $2$ for $S_M$. 
\\
Now applying Theorem \ref{thm:PIR-general} to $B$ and $\Fam$ (with $e = 2$), we get a PIR scheme with communication 
$O(k + k^{e-1}) = O(k) = \exp(\Otilde( ( \log n)^{\frac{1}{r+1}}))$.
\end{proof}

With the main theorem in hand, we now give a proof of the main corollary which simply plugs in the best known sparse $S$-decoding polynomials.
\begin{corollary1}
 For each $r \geq 2$, there is a $t$-server PIR with communication
    $\exp(\Otilde(\log n)^{\frac{1}{r+1}})$ according to the following table:
    \begin{center}
\begin{tabular}{ |c|c| } 
 \hline
 r & t  \\ 
 \hline
 $r = 1$ & 2\\
 $r$ even, $r \leq 102$ & $\left(\sqrt{3}\right)^r$ \\ 
 $r$ odd, $3 \leq r\leq 103$  & $8 \cdot \left(\sqrt{3} \right)^{r-3} $ \\ 
 $r \geq 104$  &  $\left( 3/4 \right)^{51} \cdot 2^r$ \\ 
 \hline
\end{tabular}
\end{center}
\end{corollary1}

\begin{proof}
For every positive
$r$, 
by using Theorem \ref{thm:decoding-poly}, we get
an
$S_m$-decoding polynomial over a field of characteristic $2$ with 
$t$ monomials
for the corresponding 
$t$ as in the table and 
$m$
product of 
$r$
distinct odd primes. Now by applying Theorem \ref{thm:main}, to $p =2, m$, we get the result.
\end{proof}

\bibliographystyle{alpha}
\bibliography{pir}

\end{document}